\documentclass[12pt]{amsart}

\usepackage{amssymb}
\usepackage{MnSymbol}

\usepackage[all]{xy}
\usepackage{tikz-cd}
\usetikzlibrary{arrows}

\newcommand{\mc}{\mathcal}

\newcommand{\ts}{\textsc}

\def\int{{\sf int}}

\newtheorem{thm}{Theorem}[section]

\newtheorem{cor}[thm]{Corollary}
\newtheorem{prop}[thm]{Proposition}

\theoremstyle{definition}
\newtheorem{defn}[thm]{Definition}
\newtheorem{rem}[thm]{Remark}
\newtheorem{ex}[thm]{Example}
\newtheorem{problem}{Problem}

\begin{document}

\title{Logical aspects of quantum structures}

\author{J. Harding}
\address{New Mexico State University}
\email{hardingj@nmsu.edu}
\thanks{The first listed author is partially supported by US Army grant W911NF-21-1-0247.}

\author{Z. Wang}
\address{Microsoft Station Q and Dept. of Mathematics, University of Californian at Santa Barbara, Santa Barbara, CA 93106}
\email{zhenghwa@microsoft.com}
\thanks{The second listed author is partially supported by NSF grants
FRG-1664351, CCF 2006463, and DOD Muir grant ARO W911NF-19-S-0008.}

\date{}

\subjclass[2010]{}

\begin{abstract}
We survey several problems related to logical aspects of quantum structures. In particular, we consider problems related to completions, decidability and axiomatizability, and embedding problems. The historical development is described, as well as recent progress and some suggested paths forward.
\end{abstract}

\maketitle

\section{Introduction}
This note takes an overview of a number of problems related to logical aspects of quantum structures. The quantum structures we consider are motivated by the ortholattice ${\mc P}(\mc{H})$ of projection operators of a Hilbert space $\mc{H}$. These include orthomodular lattices and orthomodular posets on the more general end of the spectrum, as well as more specialized structures such as projection lattices of finite-dimensional Hilbert spaces and ortholattices of projections of von Neumann algebras.

The problems we consider are largely based on our personal experience and interests, and represent only a fragment of the the subject. Some have a long history, including the completion problem and word problem for orthomodular lattices. The intension is to provide a survey of old results, more recent developments, and some open problems. There are a few novel contributions here, but the intent is to provide easy access to areas we feel are deserving of further attention. For further reading on some of the topics discussed here, and related topics, the reader may consult \cite{DunnIntro,Fritz,Herrmann}.

The second section provides a brief background and our perspective on quantum \mbox{structures}. The third section discusses completions, the fourth section deals with matters related to decidability and axiomatizability, and the fifth section discusses embedding problems.

\section{Background}

In their 1936 paper \cite{BvN}, Birkhoff and von Neumann noted that the closed subspaces $\mc{C}(\mc{H})$ of a Hilbert space $\mc{H}$ form a lattice with an additional unary operation $\perp$ where $A^\perp$ is the closed subspace of vectors orthogonal to those in $A$. They proposed that this lattice serve as a type of non-distributive ``logic'' for a calculus of propositions involving quantum mechanical events. Husimi \cite{Husimi} noted that $\mc{C}(\mc{H})$ satisfies the identity $A\subseteq B\Rightarrow B=A\vee(A^\perp \wedge B)$ now known as the {\em orthomodular law}. This led to the study of orthomodular lattices (see \cite{Kalmbach}).

\begin{defn}
An {\em ortholattice} (abbreviated: \textsc{ol}) $(L,\wedge,\vee,0,1,')$ is a bounded lattice with a period two order-inverting unary operation $'$ that satisfies $x\wedge x'=0$ and $x\vee x'=1$. An {\em orthomodular lattice} (abbreviated: \textsc{oml}) is an \textsc{ol} that satisfies $x\leq y \Rightarrow y=x\vee(x'\wedge y)$.
\end{defn}

There were other reasons for the interest of Birkhoff and von Neumann in the lattice $\mc{C}(\mc{H})$. Birkhoff \cite{BirkhoffPG} and Menger \cite{Menger} had recently developed the lattice-theoretic view of projective geometry. This spurred von Neumann's development of his ``continuous geometry'' \cite{ContGeo} which occurred in parallel to his work with Murray \cite{MvN1,MvN2,MvN4,MvN3} on ``rings of operators'', subalgebras of the algebra $\mc{B}(\mc{H})$ of bounded operators on $\mc{H}$ that are closed under the weak operator topology. In modern terms, these rings of operators are known as {\em von Neumann algebras}. Due to the bijective correspondence between closed subspaces of $\mc{H}$ and self-adjoint projection operators of $\mc{H}$, the \textsc{oml} $\mc{C}(\mc{H})$ is isomorphic to the \textsc{ol} $\mc{P}(\mc{H})$ of self-adjoint projection operators of $\mc{H}$. In fact, for any von Neumann algebra $\mc{A}$, its self-adjoint projections $\mc{P}(\mc{A})$ form an \textsc{oml} that comes remarkably close to determining the structure of $\mc{A}$ \cite{Dye,Hamhalter}.

It is difficult to piece together the full motivation of von Neumann during this period of amazing activity. There are his published papers and notes \cite{vNcollected}, accounts of Halperin of various issues in places such as his foreword to \cite{ContGeo}, and work of R\'edei on the subject \cite{Redei}. It is fair to say that von Neumann had a mix of logical, geometric, as well as probabilistic and measure-theoretic motivations that were never completely implemented due to the onset of the second world war.

The view that the lattice $\mc{C}(\mc{H})$, or its incarnation as projections $\mc{P}(\mc{H})$, plays a key role in quantum theory seems to have been completely born out by subsequent events. Gleason showed \cite{Anatolij} that the states of a quantum system modeled by $\mc{H}$ correspond to $\sigma$-additive measures on $\mc{P}(\mc{H})$; Ulhorn's formulation of Wigner's theorem \cite{Ulhorn} characterizes projective classes of unitary and anti-unitary operators on $\mc{H}$ as automorphisms of $\mc{P}(\mc{H})$; and of course the spectral theorem describes self-adjoint operators on $\mc{H}$ as $\sigma$-additive homomorphisms from the Borel sets of the reals to $\mc{P}(\mc{H})$.

In a different vein, Mackey \cite{Mackey} took the task of motivating the structure $\mc{P}(\mc{H})$ from simple physically meaningful assumptions. He began with abstract sets $\mc{O}$ of observables, and $\mc{S}$ of states, and used $\mc{B}$ for the Borel subsets of the reals. He assumes a function $p:\mc{O}\times\mc{S}\times\mc{B}\to [0,1]$ where $p(A,\alpha,E)$ is the probability that a measurement of observable $A$ when the system is in state $\alpha$ yields a result in the Borel set $E$. He defines the set $\mc{Q}$ of {\em questions} to be those observables $A$ that take only two values 0,1, i.e. with $p(A,\alpha,\{0,1\}) = 1$ for each state $\alpha$. From minimal assumptions, he shows that $\mc{Q}$ has the structure of what we now call an orthomodular poset (see below). He then makes the quantum leap to require that the orthomodular poset associated to a quantum system is the projection lattice $\mc{P}(\mc{H})$ of a Hilbert space.

\begin{defn}
An {\em orthocomplemented poset} (abbreviated: \textsc{op}) $(P,\leq,',0,1)$ is a bounded poset with additional period two order-inverting operation $'$ so that $x\wedge x'=0$ and $x\vee x'=1$. For elements $x,y\in P$, we say that $x,y$ are {\em orthogonal} and write $x\perp y$, if $x\leq y'$.
\end{defn}

In this definition we use the meet and join symbols to indicate that the elements have a meet or join, and to also express to what this meet or join is equal.

\begin{defn}
An {\em orthomodular poset} (abbreviated: \textsc{omp}) is an \textsc{op} where every pair of orthogonal elements have a join, and that satisfies $x\leq y\Rightarrow y=x\vee (x\vee y')'$. An \textsc{omp} is called $\sigma$-complete if every countable set of pairwise orthogonal elements has a join.
\end{defn}

The area of study that uses the lens of projection lattices, or more generally \textsc{oml}s and \textsc{omp}s, to motivate and study quantum foundations is often known as ``quantum logic''. We note that Varadarajan \cite{Varadarajan} uses ``geometric quantum theory'' for this study. There is also a notion of ``quantum logic'' much more closely aligned to traditional algebraic logic, based on the idea of replacing the Boolean algebras used in classical logic with some particular type of quantum structure such as projection lattices $\mc{P}(\mc{H})$, or general \textsc{oml}s. There are different flavors of this, see for example \cite{Weaver}.

\section{The completion problem}

We group our first set of problems under the banner of ``the completion problem.'' This consists of several different problems with an obviously similar theme. While they appear quite similar, there may turn out to be substant differences in detail.

\begin{problem}
Can every \textsc{oml} be embedded into a complete \textsc{oml}?
\end{problem}

\begin{problem}
Can every \textsc{oml}/\textsc{omp} be embedded into a $\sigma$-complete \textsc{oml}/\textsc{omp}?
\end{problem}

In considering an \textsc{omp} or \textsc{oml} as a model for the events of a quantum system, it is natural to consider $\sigma$-completeness. This is directly interpreted as providing an event comprised from a countable collection of mutually exclusive events, and makes analysis using conventional techniques from probability theory tractable. Completeness implies $\sigma$-completeness, so serves the same purpose, but it has less physical motivation. Indeed, it is difficult to motivate even the existence of binary joins and meets of non-compatible events. But completeness is used in logical applications where it provides a means to treat quantifiers, so is also of interest for this reason. It is known that there are \textsc{omp}s that cannot be embedded into an \textsc{oml} (see Section~\ref{embedding}), so these questions may have different content. We begin with the following result of MacLaren \cite{MacLaren}.

\begin{thm}
The MacNeille completion of an \textsc{ol} naturally forms an \textsc{ol}.
\end{thm}

Several core results about completions are consequences of deep early results from the study of \textsc{oml}s. For the first, recall that the orthogonal of a subspace $S$ of an inner product space $V$ is given by $S^\perp=\{v:s\cdot v = 0\mbox{ for all }s\in S\}$, and that $S$ is biorthogonal if $S=S^{\perp\perp}$. Amemiya and Araki \cite{Amemiya-Araki} provided the following influential result viewed in part as justification of the use of Hilbert space in quantum mechanics.

\begin{thm}
The \textsc{ol} of biorthogonal subspaces of an inner product space is orthomodular iff the inner product space is complete, i.e. is a Hilbert space.
\end{thm}

\begin{cor}\label{zx}
The MacNeille completion of an \textsc{oml} need not be an \textsc{oml}.
\end{cor}

\begin{proof}
Let $L$ be the \textsc{ol} of finite and co-finite dimensional subspaces of an incomplete inner product space $V$. This is a modular \textsc{ol}, hence an \textsc{oml}.  Since $L$ is join and meet dense in the \textsc{ol} $\overline{L}$ of biorthogonal subspaces of $V$ and $\overline{L}$ is complete, we have that $\overline{L}$ is the MacNeille completion of $L$.
\end{proof}

\begin{rem}\label{bn}
The \textsc{oml} $L$ of Corollary~\ref{zx} can be embedded into a complete \textsc{oml}. Embed $V$ into a complete inner product space $\overline{V}$. Then $L$ embeds into the subalgebra of finite or co-finite dimensional subspaces of $\overline{V}$, and this is a subalgebra of the \textsc{oml} of all closed subspaces of $\overline{V}$.
\end{rem}

A deep result of Kaplansky \cite{Kaplansky} settles negatively the situation for modular ortholattices (abbreviated: \textsc{mol}s). We recall that every \textsc{mol} is an \textsc{oml} but not conversely. For space, we will not provide the details of continuous geometries, the reader can see \cite{ContGeo}. The crucial fact we use is that a continuous geometry $L$ has a dimension function $d:L\to [0,1]$ that satisfies among other conditions $d(x\vee y) = d(x) + d(y) - d(x\wedge y)$.

\begin{thm}
A complete, directly irreducible \textsc{mol} is a continuous geometry.
\end{thm}

\begin{cor}\label{mkk}
There is a \textsc{mol} that cannot be embedded into a complete \textsc{mol}.
\end{cor}

\begin{proof}
Let $L$ be the \textsc{ol} of all subspaces $A$ of a Hilbert space where $A$ or $A^\perp$ is finite-dimensional. It is well-known that lattice operations with such subspaces are obtained via intersection and sum of subspaces, so this is a \textsc{mol}. The atoms of $L$ are 1-dimensional subspaces. It is easily seen that any two atoms have a common complement, i.e. are perspective. So there is an infinite set of pairwise orthogonal pairwise perspective elements in $L$. But a continuous geometry cannot have such a set of elements since they would all have the same non-zero dimension because of being perspective.
\end{proof}

So \textsc{mol}s do not admit completions, and for \textsc{oml}s the MacNeille completion does not always remain within \textsc{oml}. To gain a better understanding of the situation for \textsc{oml}s, we can further limit expectations.

\begin{defn}
An order embedding of posets $\varphi:P\to Q$ is {\em regular} if it preserves all existing joins and meets.
\end{defn}

The following result is found in \cite{hardingcanonical}, and is based on a result of Palko \cite{Palko}. The proof of the result extends in an obvious way to other situations. For instance, it shows that there is an \textsc{omp} that cannot be regularly embedded into a $\sigma$-complete \textsc{omp}.

\begin{thm}
A regular completion of an \textsc{oml} factors as a pair of regular embeddings through the MacNeille completion.
\end{thm}

\begin{cor}
There is no regular completion for \textsc{oml}s.
\end{cor}

\begin{rem}
The examples discussed so far have all involved the \textsc{oml} of subspaces of a Hilbert space. There is an alternate source of examples of interest for the completion problem. A construction of Kalmbach \cite{Kalmbach} builds from a bounded poset $P$ an \textsc{omp} $K(P)$. In the case that $P$ is a lattice, $K(P)$ is an \textsc{oml}. This construction works by {\em gluing} the free Boolean extensions of chains of $P$. The forgetful functor ${\sf U}:{\sf Omp} \to {\sf Pos}$ takes the category of orthomodular posets and maps that preserve orthocomplementation and finite orthogonal joins to the category of bounded posets and order-preserving maps. The Kalmbach construction provides an adjoint ${\sf K}:{\sf Pos}\to{\sf Omp}$ to the forgetful functor \cite{HardKalmbach}. Effect algebras are the Eilenberg-Moore category over this monad \cite{Jenca}.

There is a simple condition on a lattice $L$ equivalent to the MacNeille completion of $K(L)$ being an \textsc{oml}. This provides a rich source of relatively transparent examples. For any lattice completion $C$ of $L$, we have that the MacNeille completion of $K(C)$ is an \textsc{oml}, and this provides a completion of $K(L)$. Again, as in Remark~\ref{bn}, a completion of an \textsc{oml} is obtained by first completing some underlying structure.
\end{rem}

We turn our attention to a method of completing \textsc{ol}s that is far from regular in that it destroys all existing joins except those that are essentially finite. This is termed the canonical completion. The canonical completion has a long history originating with completing Boolean algebras with operators. Here, the canonical completion of a Boolean algebra $B$ is realized as the embedding into the power set of the Stone space of $B$. For its application to \textsc{ol}s, see \cite{MaiJohn,hardingcanonical}.

\begin{defn}
For an \textsc{ol} $L$, we say an embedding of $L$ into a complete \textsc{ol} $L^\sigma$ is a canonical completion of $L$ if every element of $L^\sigma$ is a join of meets of elements of $L$ and for each $S,T\subseteq L$, if $\bigwedge S\leq \bigvee T$, then there are finite $S'\subseteq S$ and $T'\subseteq T$ with $\bigwedge S'\leq\bigvee T'$.
\end{defn}

\begin{thm}
Every \textsc{ol} has a canonical completion, and this is unique up to unique commuting isomorphism.
\end{thm}

Canonical completions are better at preserving algebraic properties than MacNeille completions. In \cite{MaiJohnYde} it is shown that any variety that is closed under MacNeille completions is closed under canonical completions. However, canonical completions do not provide the answers we seek here.

\begin{thm}\label{vbm}
The canonical completion of an \textsc{mol} need not be an \textsc{mol} and the canonical completion of an \textsc{oml} need not be an \textsc{oml}.
\end{thm}

\begin{proof}
The first statement follows from Corollary~\ref{mkk}. The second is established in \cite{hardingcanonical} where necessary and sufficient conditions are given for the canonical completion of $K(L)$ to be an \textsc{oml}.
\end{proof}

We have occasion now to discuss matters related to states. We recall that a {\em state} on an \textsc{omp} $P$ is a map $\sigma:P\to [0,1]$ that preserves bounds and is {\em finitely additive} meaning that $x\perp y$ implies $\sigma(x\vee y) = \sigma(x)+\sigma(y)$. {\em Pure states} are ones that cannot be obtained as a non-trivial convex combination of others.

\begin{rem}
The proof of Theorem~\ref{vbm} gives more. Each $K(L)$ has a full set of 2-valued states, i.e. states taking only values 0,1, hence is what is known as a {\em concrete} \textsc{oml}. The concrete \textsc{oml}s form a variety, and the above results show that this variety is not closed under MacNeille completions or canonical completions.
\end{rem}

Having given a number of negative results, we mention a direction that produces strong positive results in a physically motivated setting. The first result in this area was by Bugajska and Bugajski \cite{Bugajski}. Here we follow a sequel to this result by Guz \cite{Guz} which we reformulate below.

\begin{thm}\label{nb}
Let $P$ be an \textsc{omp} and $\mc{S}$ be its set of pure states. Suppose $P$ satisfies
\begin{enumerate}
\item For each non-zero $x\in P$ there is $\sigma\in\mc{S}$ with $\sigma(x)=1$
\item If $x\nleq y$ then there is $\sigma\in\mc{S}$ with $\sigma(x)=1$ and $\sigma(y)\neq 1$
\item For each $\sigma\in\mc{S}$ there is $x\in P$ with $\sigma(x)=1$ and $\sigma'(x)\neq 1$ for all $\sigma'\in\mc{S}$ with $\sigma\neq\sigma'$.
\end{enumerate}
Then $P$ is atomistic and its MacNeille completion is an \textsc{oml}.
\end{thm}

\begin{rem}
Theorem~\ref{nb} is intended to provide simple physical assumptions on the events of a quantum system that ensure it can be embedded into a complete \textsc{oml}. The assumptions here have physical interpretation. The first says that each event is certain in some pure state. The second says that $x\leq y$ iff whenever $x$ is certain in some state, then so also is $y$. Guz describes the third as saying that for any pure state, there is an event that can be used to test for it. The line of reasoning is certainly of interest, but the axioms chosen are not without issue. The third axiom in particular is quite strong. For instance, the axioms given will not hold in a non-atomic Boolean algebra. It is of interest to see if less restrictive conditions on the set of states provides similar results.
\end{rem}

For the remainder of this section, we discuss some possible directions for approaches to completion problems. These are speculative, and may not turn out to be of use. But they seem worthy of further consideration. We begin with the following result of \cite{BrunsHardingAmal} called the Boolean amalgamation property.

\begin{thm}\label{gh}
Let $L_1, L_2$ be \textsc{oml}s that intersect in a common Boolean subalgebra. Then there is an \textsc{oml} $L$ containing $L_1,L_2$ as subalgebras.
\end{thm}

\begin{rem}
As with any lattice, an \textsc{oml} $L$ is complete iff each chain in $L$ has a join. Since a chain in an \textsc{oml} is contained in a {\em block} (a maximal Boolean subalgebra), an \textsc{oml} is complete iff all of its blocks are complete. The Boolean amalgamation property lets us complete a given block $B$ of an \textsc{oml} $L$. Taking a completion $C$ of $B$, we have that $L$ and~$C$ intersect in common Boolean subalgebra $B$, so there is an \textsc{oml} $M$ containing $L$ and $C$ as subalgebras. In effect, $B$ has been completed within $M$.

One could hope to iterate this process to obtain a $\sigma$-completion of $L$. Unfortunately, such an iterative approach requires that we preserve joins completed at an early stage, and the Boolean amalgamation described in Theorem~\ref{gh} does not do this. Perhaps there is a modification of the proof of Theorem~\ref{gh} that does allow this. On the other hand, Theorem~\ref{gh} has implications in terms of constructing an example of an \textsc{oml} that cannot be completed to an \textsc{oml}. Roughly, one particular join can always be inserted.
\end{rem}

Our next topic, due to Rump \cite{Rumpvn}, is a translation of orthomodular lattices into structures belonging to classical algebra. In \cite{Unigroup} Foulis, Greechie and Bennett associated to an \textsc{oml}~$L$ its {\em unigroup}. This unigroup $G(L)$ is a partially ordered abelian group with strong order unit $u$ and map $\mu_L$ from $L$ into the interval $[0,u]$ of $G$. The map $\mu_L:L\to G(L)$ is universal among abelian group-valued measures, meaning that for any abelian group-valued measure $\mu:L\to G$ there is a group homomorphism $f_\mu:G(L)\to G$ with $f\circ\mu_L = \mu$. However, $\mu_L$ need not be an embedding. Rump repairs this defect by extending to the setting of non-commutative groups.

\begin{defn}
A {\em right $\ell$-group} is a group $G$ equipped with a lattice structure $\leq$ such that $a\leq b\Rightarrow ac\leq bc$. An element $u$ is a {\em strong order unit} if $b\leq c\Leftrightarrow ub\leq uc$ for all $b,c\in G$ and each $a\in G$ satisfies $a\leq u^n$ for some natural number $n$. A strong order unit $u$ is called {\em singular} if $u^{-1}\leq xy \Rightarrow yx=x\wedge y$.
\end{defn}

Rump called a right $\ell$-group with singular strong order unit an {\em orthomodular group}. Using a construction that extends Foulis' construction of a Baer *-semigroup from an \textsc{oml} \cite{FoulisBaer*}, he provided a method to associate to any \textsc{oml} $L$ an orthomodular group $G_R(L)$, called the structure group of $L$. We denote this $G_R(L)$ to distinguish it from the unigroup of $L$. However, it has the same universal property for group-valued measures as the unigroup has for abelian group-valued measures. Moreover, he shows the following \cite[Th,.~4.10]{Rumpvn}.

\begin{thm}
If $G$ is an orthomodular group with singular strong order unit $s$, then the interval $[s^{-1},1]$ is an \textsc{oml} with structure group $G$. Conversely, each \textsc{oml} arises this way.
\end{thm}

As is usual with ordered groups, we say an orthomodular group is complete iff every bounded subset has a least upper bound. This property is often known as Dedekind complete. Rump showed that an \textsc{oml} is complete iff its structure group is complete. This leads to the question of whether techniques from the theory of $\ell$-groups can be used to find a completion for \textsc{oml}s. Of course, the results of Foulis on Baer *-semigroups allowed a similar algebraic path for many decades without result, but there has been much progress in the study of $\ell$-groups.

\section{Decidability and axiomatizability} \label{jj}

Given a class of algebras $\mc{K}$, several logical questions arise. One can ask if there is an algorithm to decide the equational theory $\sf{Eq}(\mc{K})$ or first order theory $\sf{Th}(\mc{K})$, and if there is a finite set of equations, respectively first order formulas, that axiomatize these theories. When $\mc{K}$ consists of a single algebra $A$ we write $\sf{Eq}(A)$ and $\sf{Th}(A)$. We begin with some standard terminology.

\begin{defn}
A variety $V$ has {\em solvable free word problem} if there is an algorithm to determine if an equation holds in $V$.
\end{defn}

Since $V$ has the same equational theory as its free algebra $F_V(\omega)$ over countably many generators, solving the free word problem for $V$ amounts to giving an algorithm to decide if an equation holds in its countably generated free algebra. Before discussing the situation for specific varieties, we describe a general technique.

\begin{defn}
A {\em partial subalgebra} of an algebra $A$ is a subset $S\subseteq A$ equipped with partial operations being the restriction of the operations of $A$ to those tuples in $S$ where the result of the operation in $A$ belongs to $S$. A variety $V$ has the {\em finite embedding property} if each finite partial subalgebra of an algebra in $V$ is a partial subalgebra of a finite algebra in $V$.
\end{defn}

There are a number of connections between partial algebras and word problems, see for example \cite{Evans,Sapir}. The following is obtained via a back and forth argument of checking for a proof of an equation and looking for a counterexample among finite algebras.

\begin{thm}\label{fep}
If a variety $V$ is finitely axiomatized and has the finite embedding property, then it has solvable free word problem.
\end{thm}

We begin with two varieties at the opposite ends of the spectrum of quantum structures, the variety \ts{ba} of Boolean algebras and \ts{ol}. As is often the case where there is a great deal of structure, or relatively little structure, we have solvable free word problems.

\begin{thm}
The free word problems in \ts{ba} and in \ts{ol} are solvable.
\end{thm}

In both cases the result follows easily from Theorem~\ref{fep}. The variety \ts{ba} has the finite embedding property since every finitely generated \ts{ba} is finite. The variety \ts{ol} has the finite embedding property since the MacNeille completion of an orthocomplemented poset is an \ts{ol} \cite{MacLaren} and MacNeille completions preserve all existing joins and meets. In both cases we can give a much more tractable algorithm to decide the free word problem than that provided by Theorem~\ref{fep}. Since \ts{ba} is generated by the 2-element \ts{ba} we need only check validity of an equation in $2$, essentially the method of truth tables. For \ts{ol} Bruns \cite{BrunsFreeOL} gave an explicit algorithm based on Whitman's algorithm for free lattices \cite{Freese}.

\begin{problem}
Is the free word problem for the variety \ts{oml} solvable?
\end{problem}

This problem has received considerable attention over the years, but without a great deal of progress. Trying to establish the finite embedding property for \ts{oml}s was a motivation behind some of the work of Bruns and Greechie on commutators, and the idea of Kalmbach's attempt at a solution to the free word problem \cite{KalmbachFree} that seems to have a gap. There are other hopes to solve the free word problem for \ts{oml} that proceed by finding some other free word problem whose solution would yield that of \ts{oml}. This is the direction of the recent work of Fussner and St. John \cite{Fussner} involving ortholattices where a derived operation is residuated.

We turn briefly to a more general discussion of free algebras in \ts{ba}, \ts{ol}, and \ts{oml}. The following collects a number of known results. We recall, that for a cardinal $\kappa$, that MO$_\kappa$ is the \ts{mol} of height 2 with a bottom, a top, and $\kappa$ pairs of incomparable orthocomplemtary elements in the middle.

\begin{thm}
In \ts{ol} the free algebra on 2 generators contains the free one on countably many generators as a subalgebra. In \ts{oml} the free algebra on 2 generators is MO$_2\times 2^4$ and the free algebra on 3 generators contains the free one on countably many generators as a subalgebra.
\end{thm}

The results about free \ts{ol}s are due to Bruns \cite{BrunsFreeOL}. The result about the free \ts{oml} on 2 generators is due to Beran \cite{Kalmbach}. It has a significant impact on the study of \ts{oml}s since it makes calculations involving 2 variables tractable. The result about the free \ts{oml} on 3 generators containing the countably generated free one is due to Harding \cite{HardingFree}. The situation for \ts{ol} is similar to that of lattices with very similar algorithms providing a solution to the free word problem. Yet, while there is an extensive literature on properties of free lattices \cite{Freese}, relatively little is known about the structure of free \ts{ol}s.

\begin{problem}
Obtain a better understanding of the structure of free \ts{ol}s and free \ts{oml}s. In particular, what are their finite subalgebras? In a free \ts{ol}, if $b$ is a complement of $a$, are $a'\vee b$ and $a'\wedge b$ also complements of $a$? Can a free \ts{oml} contain an uncountable Boolean subalgebra?
\end{problem}

The situation for \ts{mol}s also leaves a great deal open. The free \ts{mol} on 2 generators is also MO$_2\times 2^4$ since this is modular and is free on 2 generators in \ts{oml}. The free \ts{mol} on 3 generators is infinite. We are not aware of whether it contains the free one on countably many generators as a subalgebra. Roddy has provided a finitely presented \ts{mol} with unsolvable word problem \cite{Roddy}. Every finite subdirectly irreducible \ts{mol} is in the variety generated by MO$_\omega$ \cite{BrunsMOL}, hence satisfies the 2-distributive law. Thus \ts{mol} is not generated by its finite members and so cannot have the finite embedding property. This leaves the following open problem.

\begin{problem}
Does \ts{mol} have solvable free word problem?
\end{problem}


Dunn, Moss, and Wang in their ``Third life of quantum logic'' \cite{DunnIntro} pointed to the value of studying free word problems for the algebras most tightly tied to quantum computing. They used $\sf{QL}(\mathbb{C}^n)$ for the equational theory of the \ts{mol} of closed subspaces $\mc{C}(\mathbb{C}^n)$, and they called this the {\em quantum logic} of $\mathbb{C}^n$. We extend this practice and use $\sf{QL}(\mc{R})$ for the equational theory of the projection lattice of an arbitrary type II$_1$ factor $\mc{R}$, and $\sf{QL}(CG(\mathbb{C}))$ for the equational theory of the  orthocomplemented continuous geometry $CG(\mathbb{C})$ constructed by von Neumann \cite{ContGeo} via a metric completion of a limit of subspace lattices. We summarize results obtained in \cite{Dunn,HardingContGeo,Herrmann} below. In particular, the note of Herrmann \cite{HerrmannNote} is an excellent description of the situation.

\begin{thm}\label{vbv}
$\sf{QL}(\mathbb{C}) \supset \sf{QL}(\mathbb{C}^2)$ $\supset \cdots \supset \bigcap\{\sf{QL}(\mathbb{C}^n):n\geq 1\} = \sf{QL}(CG(\mathbb{C}))=\sf{QL}(\mc{R})$ for each type II$_1$ factor $\mc{R}$. Each of these containments is strict. Each of these equational theories is decidable, and the first order theory of each \textsc{mol} $\mc{C}(\mathbb{C}^n)$ for $n\geq 1$ is decidable.
\end{thm}

The containments among the $\sf{QL}(\mathbb{C}^n)$ and $\bigcap\{\sf{QL}(\mathbb{C}^n):n\geq 1\}$ are trivial. That they are strict follows from the fact \cite{Huhn} that $\mc{C}(\mathbb{C}^k)$ is $n$-distributive iff $k\leq n$. The two equalities are established in \cite{Herrmann}. Decidability of the first order theory of each $\mc{C}(\mathbb{C}^n)$ is given in \cite{Dunn} by translating formulas of the \ts{mol} to formulas about $\mathbb{C}$, and using Tarski's theorem \cite{Tarski2} on the decidability of the first order theory of $\mathbb{C}$. This has as a consequence the decidability of the equational theories $\sf{QL}(\mathbb{C}^n)$ for $n\geq 1$. The decidability of the equational theories in the remaining cases was established in \cite{Herrmann}, and independently for $\sf{QL}(CG(\mathbb{C}))$ in \cite{HardingContGeo}.

\begin{defn}
A {\em quasi-equation} is a formula $\forall x_1,\ldots,\forall x_n(s_1=t_1\, \&\, \cdots\, \&\, s_k=t_k\,\longrightarrow\, s=t)$ where $s_1=t_1,\ldots,s_k=t_k$ and $s=t$ are equations. The {\em uniform word problem} for a variety $V$ asks whether there is an algorithm that determines which quasi-equations are valid in $V$.
\end{defn}

The uniform word problem asks if there is a single algorithm that decides when two words are equal for a finitely presented algebra in $V$. Of course, each equation is equivalent to a quasi-equation, so it is (much) more difficult to have a positive solution to the uniform word problem than to have a positive solution to the free word problem.

\begin{thm}
\ts{mol} has unsolvable uniform word problem, as does the variety generated by the projection lattice of any type II$_1$ factor and the variety generated by $CG(\mathbb{C})$.
\end{thm}

The first statement was shown by Roddy, who gave a finitely presented \ts{mol} with unsolvable word problem. The second statement is in \cite{Herrmann}.

\begin{prop}
The first order theory of $\mc{C}(\mathbb{C}^n)$ is finitely axiomatizable iff $n=1$. The first order theory of the closed subspaces $\mc{C}(\mc{H})$ of an infinite-dimensional Hilbert space is not finitely axiomatizable.
\end{prop}

Since $\mc{C}(\mathbb{C})$ is the 2-element Boolean algebra, the case $n=1$ is trivial. For $n\geq 3$ one can recover the field $\mathbb{C}$ from the \ts{mol} $\mc{C}(\mathbb{C}^n)$ by the standard lattice-theoretic treatment of the usual techniques from projective geometry (see eg. \cite{Crawley,Faure}). Since this process is first order, a finite axiomatization of the first order theory of $\mc{C}(\mathbb{C}^n)$ would give a finite axiomatization of the first order theory of the field $\mathbb{C}$. This is not possible since any sentence true in $\mathbb{C}$ is true in algebraically closed fields of sufficiently large prime characteristic (see. eg. \cite[p.~2]{Marker}). For further discussion of the case $n=3$, see \cite{HerrmannNote}.

For the case $n=2$, note that $\mc{C}(\mathbb{C}^2)$ is MO$_\mathfrak{c}$ where $\mathfrak{c}$ is the continuum. There is a non-principle ultraproduct of the MO$_n$ for $n\in\mathbb{N}$ of cardinality $\mathfrak{c}$ \cite{Frayne}, and by \L o\'{s}'s Theorem, this ultrapower must be MO$_\mathfrak{c}$. If $\sf{Th}($MO$_\mathfrak{c})$ can be finitely axiomatized, then it can be axiomatized by a single sentence $\varphi$. But them each MO$_n\models\neg\varphi$ and by \L o\'{s}'s theorem, the ultraproduct MO$_\mathfrak{c}\models\neg\varphi$, an impossibility.

If $\mc{H}$ is infinite-dimensional, then $\mc{C}(\mc{H})$ has an element $p$ of height 3, and $[0,p]$ is isomorphic to $\mc{C}(\mathbb{C}^3)$. The result follows from the previous ones.

\begin{rem}
At this point, there are many questions related to this line of investigation. Several are raised in \cite{HerrmannNote}. We resist the temptation to formulate more problems related to the current material, but we will post a further problem raised in \cite{DunnIntro}.
\end{rem}

\begin{problem}\label{CH}
Is the equational theory of the \ts{oml} $\mc{C}(\mc{H})$ decidable?
\end{problem}

This problem is one way to address what has been a primary issue since the early days of quantum logic, understanding more deeply the \ts{ol} $\mc{C}(\mc{H})$. Birkhoff and von Neumann \cite{BvN} knew that $\mc{C}(\mc{H})$ did not belong to \ts{mol}. Husimi \cite{Husimi} formulated the orthomodular law that separated the equational theory of $\mc{C}(\mc{H})$ from that of \ts{ol}. Day introduced the ``ortho-Arguesian law'', an equation in six variables related to the Arguesian condition of projective geometry. He showed that this equation is valid in $\mc{C}(\mc{H})$ and not in \ts{oml}. Many refinements of this ortho-Arguesian identity have been found \cite{Megil1,Megil2} providing other equations valid in $\mc{C}(\mc{H})$ and not in \ts{oml}. A further source of equations valid in $\mc{C}(\mc{H})$ and not in \ts{oml} is provided by the fact that $\mc{C}(\mc{H})$ has an ample supply of well-behaved states \cite{GreechieStates,Mayet1,Mayet2}.

In a different direction, Fritz \cite{Fritz} used results of Slofstra \cite{Slofstra} from combinatorial group theory to establish the following.

\begin{thm}
For an infinite-dimensional Hilbert space $\mc{H}$, the uniform word problem for the variety generated by $\mc{C}(\mc{H})$ is unsolvable.
\end{thm}

The result shown is actually quite a bit more specific than this, showing that there is no decision procedure for quasi-equations of a very specific form. These quasi-equations encode when certain configurations can be embedded into $\mc{C}(\mc{H})$, and a discussion of them naturally leads to the topic of our next section.




\section{Embedding problems} \label{embedding}

An embedding of one \ts{oml} into another is a one-one \ts{ol} homomorphism. Our first problem is formulated an an open-ended fashion, but accurately reflects the intent.

\begin{problem}
Increase our understanding of which \ts{oml}s can be embedded into $\mc{C}(\mc{H})$ for some Hilbert space $\mc{H}$.
\end{problem}

Note that each $\mc{C}(\mathbb{C}^n)$ for $n\geq 1$ embeds into $\mc{C}(\mc{H})$ for $\mc{H}$ an infinite-dimensional separable Hilbert space. So our interest primarily lies in embeddings into $\mc{C}(\mc{H})$ when $\mc{H}$ is infinite-dimensional.

The discussion at the end of Section~\ref{jj} gives a number of equations that are valid in $\mc{C}(\mc{H})$ but not valid in all \ts{oml}s. These include the ortho-Arguesian law and its variants, and also equations holding in all \ts{oml}s with a sufficient supply of certain types of states. Failure of any such equation in an \ts{oml} implies that it cannot be embedded into $\mc{C}(\mc{H})$. Also, $\mc{C}(\mc{H})$ has a strongly order determining set of states, meaning that $a\leq b$ iff each finitely additive state $s$ with $s(a)=1$ has $s(b)=1$. Any \ts{oml} without a strongly order determining set of states cannot be embedded into $\mc{C}(\mc{H})$.

For the other side of the question, it seems that relatively little is known about methods to determine that a given \ts{oml} can be embedded into $\mc{C}(\mc{H})$. The projections of a von Neumann algebra are an \ts{oml} that can be embedded into $\mc{C}(\mc{H})$, so in a sense this provides a source of examples. There is also an example \cite{BrunsRoddy} of a \ts{mol} with interesting properties constructed as a subalgebra of $\mc{C}(\mc{H})$. This example is constructed by carefully choosing bases of infinite-dimensional subspaces of $\mc{C}(\ell^2)$ and using delicate arguments. Aside from relatively simple cases that can easily be seen to embed in $\mathbb{C}(\mathbb{C}^n)$ for some $n$, we are aware of few positive results in this direction. To illustrate the situation, consider the following.

\begin{defn}\label{aqw}
The diagram below at left is the \ts{oml} constructed as the {\em horizontal sum} of the Boolean algebras $2^3$ and $2^2$ and is written $2^3\oplus 2^2$. The diagram in the middle consists of a family of $n$ copies of a 3-element Boolean algebra glued together at an atom, coatom, and 0,1 as shown. This is called an $n$-{\em element chain} because of its appearance when viewed as a Greechie diagram. The diagram at right is obtained from the one in the middle by identifying the two copies of $a$ and the two copies of $a'$. This is called an $n$-{\em loop}. This is an \ts{omp} when $n\geq 4$ and an \ts{oml} when $n\geq 5$.
\end{defn}
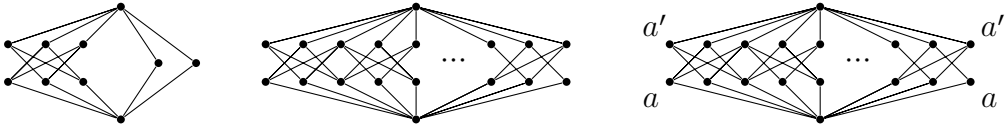
\vspace{1ex}

\begin{center}
\begin{tikzpicture}[Q/.style={circle,fill=black,inner sep=0pt,minimum size=3pt}]
\node[Q] (0) at (0,0) {}; \node[Q] (a) at (-.5,.5) {}; \node[Q] (b) at (-1,.5) {}; \node[Q] (c) at (-1.5,.5) {}; \node[Q] (a') at (-.5,1) {}; \node[Q] (b') at (-1,1) {}; \node[Q] (c') at (-1.5,1) {}; \node[Q] (1) at (0,1.5) {}; \node[Q] (d) at (.5,.75) {}; \node[Q] (d') at (1,.75) {};
\draw (0)--(a)--(b')--(c)--(a')--(b)--(c')--(1)--(b')--(c)--(0)--(b)--(a')--(1)--(d)--(0)--(d')--(1)--(c')--(a);
\end{tikzpicture}
\hspace{3ex}
\begin{tikzpicture}[Q/.style={circle,fill=black,inner sep=0pt,minimum size=3pt}]
\node[Q] (0) at (.5,0) {}; \node[Q] (a) at (-.5,.5) {}; \node[Q] (b) at (-1,.5) {}; \node[Q] (c) at (-1.5,.5) {}; \node[Q] (a') at (-.5,1) {}; \node[Q] (b') at (-1,1) {}; \node[Q] (c') at (-1.5,1) {}; \node[Q] (1) at (.5,1.5) {}; \node[Q] (d) at (0,.5) {}; \node[Q] (e) at (.5,.5) {}; \node[Q] (d') at (0,1) {}; \node[Q] (e') at (.5,1) {}; \node[Q] (g) at (1.5,.5) {}; \node[Q] (h) at (2,.5) {}; \node[Q] (i) at (2.5,.5) {}; \node[Q] (g') at (1.5,1) {}; \node[Q] (h') at (2,1) {}; \node[Q] (i') at (2.5,1) {};
\draw (0)--(a)--(b')--(c)--(a')--(b)--(c')--(1)--(b')--(c)--(0)--(b)--(a')--(1)--(c')--(a);
\draw (0)--(d)--(a')--(e)--(d')--(a)--(e')--(1)--(d')--(e)--(0); \draw (d) -- (e');
\draw (0)--(g)--(i')--(1)--(h')--(i)--(0)--(h)--(i')--(1)--(g')--(h)--(0)--(g)--(h');
\node at (1,.75) {$\cdots$};
\end{tikzpicture}
\hspace{3ex}
\begin{tikzpicture}[Q/.style={circle,fill=black,inner sep=0pt,minimum size=3pt}]
\node[Q] (0) at (.5,0) {}; \node[Q] (a) at (-.5,.5) {}; \node[Q] (b) at (-1,.5) {}; \node[Q] (c) at (-1.5,.5) {}; \node[Q] (a') at (-.5,1) {}; \node[Q] (b') at (-1,1) {}; \node[Q] (c') at (-1.5,1) {}; \node[Q] (1) at (.5,1.5) {}; \node[Q] (d) at (0,.5) {}; \node[Q] (e) at (.5,.5) {}; \node[Q] (d') at (0,1) {}; \node[Q] (e') at (.5,1) {}; \node[Q] (g) at (1.5,.5) {}; \node[Q] (h) at (2,.5) {}; \node[Q] (i) at (2.5,.5) {}; \node[Q] (g') at (1.5,1) {}; \node[Q] (h') at (2,1) {}; \node[Q] (i') at (2.5,1) {};
\draw (0)--(a)--(b')--(c)--(a')--(b)--(c')--(1)--(b')--(c)--(0)--(b)--(a')--(1)--(c')--(a);
\draw (0)--(d)--(a')--(e)--(d')--(a)--(e')--(1)--(d')--(e)--(0); \draw (d) -- (e');
\draw (0)--(g)--(i')--(1)--(h')--(i)--(0)--(h)--(i')--(1)--(g')--(h)--(0)--(g)--(h');
\node at (1,.75) {$\cdots$};
\node at (-1.75,.25) {$a$}; \node at (-1.7,1.25) {$a'$}; \node at (2.75,.25) {$a$}; \node at (2.8,1.25) {$a'$};
\end{tikzpicture}
\end{center}

We are not aware if it is known when an $n$-chain or $n$-loop can be embedded as an \ts{oml} into $\mc{C}(\mc{H})$. It would be desirable to have technology sufficient to answer such basic questions, even if it winds up being undecidable when a finite \ts{oml} can be embedded into $\mc{C}(\mc{H})$. For the \ts{oml} $2^3\oplus 2^2$, it is remarked in \cite{GreechieI} that Ramsey had shown that it could be embedded into $\mc{C}(\mc{H})$, but the result is unpublished and we know of no proof in print. We add this below.

\begin{prop}\label{kkm}
There is an embedding of the \ts{oml} $2^3\oplus 2^2$ into $\mc{C}(\mc{H})$ where $\mc{H}$ is a separable Hilbert space.
\end{prop}

\begin{proof}
Let $\mc{H}=L^2(\mathbb{R})$ be the square integrable complex functions on $\mathbb{R}$ modulo equivalence a.e. and let $\mc{F}$ be the Fourier transform and $\mc{F}^{-1}$ the inverse Fourier transform. Let $A, B, C$ be the closed subspaces of all functions vanishing a.e. on $(-\infty,-1)$, $(-1,1)$, and $(1,\infty)$ respectively; and let $D, E$ be the closed subspaces of all functions whose Fourier transforms vanish a.e. on $(-\infty,0)$ and $(0,\infty)$ respectively. Then $A,B,C$ are the atoms of an 8-element Boolean subalgebra of $\mc{C}(\mc{H})$, and $D,E$ are the atoms of a 4-element Boolean subalgebra of $\mc{C}(\mc{H})$. To establish the result, it is sufficient to show that any of $A',B',C'$ intersect with $D,E$ trivially.

Suppose $f\in D$. Since the Fourier transform $\hat{f}$ of $f$ vanishes a.e. on the negative reals, by Titchmarsh's theorem \cite[Thm.~95]{Titchmarsh} there is a holomorphic function $F$ defined on the upper half-plane so that $f(x)=\lim_{z\to x} F(z)$ a.e. If $f$ belongs to one of $A',B',C'$, then it is zero a.e. on a set of positive measure, so by the Luzin-Privolov theorem \cite{Luzin} is zero a.e. Thus $D$ intersects each of $A',B',C'$ trivially. The argument for $f\in E$ follows since for $g(x)=f(-x)$ we have $\hat{g}(\xi)=\hat{f}(-\xi)$.
\end{proof}

\begin{rem}
The proof of Proposition~\ref{kkm} shows more. It shows that for $B$ the Boolean algebra of Lebesgue measurable subsets of $\mathbb{R}$ modulo sets of measure zero, that $B\oplus 2^2$ is a subalgebra of $\mc{C}(\mc{H})$. Thus $2^n\oplus 2^2$ is a subalgebra of $\mc{C}(\mc{H})$ for each natural number $n$. One might hope that more is true, that for $B$ as described, that $B\oplus B$ is a subalgebra of $\mc{C}(\mc{H})$. This may or may not be the case, but there is a difficulty in extending the proof in the obvious way. In \cite{Kargaev} it is shown that there is a set $E\subseteq\mathbb{R}$ of positive finite measure so that the Fourier transform of its characteristic function vanishes on an interval.
\end{rem}

This line of investigation illustrates the difference between having a \ts{oml} embedding of an \ts{oml} $L$ into $\mc{C}(\mc{H})$ and having an \ts{omp} embedding of $L$ into $\mc{C}(\mc{H})$. Indeed, any horizontal sum of Boolean algebras is {\em concrete} since it has a full set of 2-valued states, and so can easily be embedded as an \ts{omp} into a Boolean algebra. So it is trivial that $2^3\oplus 2^2$ has an \ts{omp} embedding into $\mc{C}(\mc{H})$, as well as many other easy facts. This leads us to our next problem.

\begin{problem}\label{asw}
When can an \ts{omp} be embedded into $\mc{C}(\mc{H})$ for some Hilbert space $\mc{H}$?
\end{problem}

While in many cases it is easy to embed an \ts{omp} into $\mc{C}(\mc{H})$, not all finite \ts{omp}s can be embedded into $\mc{C}(\mc{H})$ since there are finite \ts{omp}s without any states. The most notable work on determining which \ts{omp}s can be embedded into $\mc{C}(\mc{H})$ is from Fritz \cite{Fritz} using work of Slofstra \cite{Slofstra} as we mentioned at the end of Section~\ref{jj}. We describe this in more detail.

\begin{defn}
Let $M$ be an $m\times n$ matrix with coefficients in $\mathbb{Z}_2$ and $b$ be a column vector of length $m$ with coefficients in $\mathbb{Z}_2$. A {\em quantum solution} to a linear equation $Mx=b$ over $\mathbb{Z}_2$ is a sequence $A_1,\ldots,A_n$ of self-adjoint bounded operators of a Hilbert space $\mc{H}$ such that
\begin{enumerate}
\item $A_i^2 = 1$ for each $i\leq n$,
\item $A_i$ and $A_j$ commute if $x_i$ and $x_j$ both appear in some equation,
\item For each each equation $x_{k_1}+\cdots +x_{k_r}=b_r$ we have $A_{k_1}\cdots A_{k_r}=(-1)^{b_r}1$.
\end{enumerate}
\end{defn}

In \cite{Slofstra} it was shown that it is undecidable whether a given linear equation over $\mathbb{Z}_2$ has a quantum solution. Frtiz \cite{Fritz} translated this into a form that begins to resemble the problem of embedding an \ts{omp} into $\mc{C}(\mc{H})$ as we now explain.

\begin{defn}
A {\em hypergraph} is a set $V$ of {\em vertices} and a collection $E\subseteq$ Pow$(V)$ of subsets of $V$ called {\em edges} such that each vertex lies in at least one edge. A {\em quantum representation} of a hypergraph is a mapping $\rho$ from the set of vertices to the projection operators of a Hilbert space $\mc{H}$ with dim$(\mc{H})>0$ such that for any edge $E$ we have $\sum_{v\in E}\, \rho(v) = 1$.
\end{defn}

\begin{rem}
Note that the condition $\sum_{v\in E}\, \rho(v)=1$ implies that $\rho(v)$ is orthogonal to $\rho(w)$ for $v\neq w$ belonging to a common edge. It is however useful to note that being orthogonal does not mean being distinct since $0$ is orthogonal to itself.
\end{rem}

\begin{ex}\label{nh}
A hypergraph with 7 vertices and 3 edges is shown below.
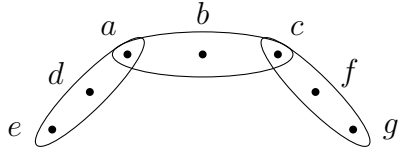
\vspace{2ex}
\begin{center}
\begin{tikzpicture}[Q/.style={circle,fill=black,inner sep=0pt,minimum size=3pt}]
\node[Q] (d) at (0,0) {}; \node[Q] (c) at (-1,0) {}; \node[Q] (e) at (1,0) {}; \node[Q] (a) at (-2,-1) {}; \node[Q] (b) at (-1.5,-.5) {}; \node[Q] (f) at (1.5,-.5) {}; \node[Q] (g) at (2,-1) {}; \draw (0,0) ellipse (1.2 and .3); \draw[rotate around = {45:(-1.5,-.5)}] (-1.5,-.5) ellipse (1 and .22);
\draw[rotate around = {-45:(1.5,-.5)}] (1.5,-.5) ellipse (1 and .22);
 \node at (-2.5,-1) {$e$}; \node at (-1.95,-.25) {$d$}; \node at (-1.25,.35) {$a$}; \node at (0,.55) {$b$}; \node at (1.25,.35) {$c$}; \node at (1.95,-.25) {$f$}; \node at (2.5,-1) {$g$};
\end{tikzpicture}
\end{center}
There are many quantum representations of this hypergraph. Let $i, j, k$ be the standard basis vectors of $\mathbb{C}^3$. Using $P_v$ for the projection onto the one-dimensional subspace spanned by the vector $v$, set
\[ \rho(a) = P_i,\hspace{1ex} \rho(b) = P_j,\hspace{1ex} \rho(c) = P_k,\hspace{1ex} \rho(d) = P_{j+k},\hspace{1ex} \rho(e) = P_{j-k},\hspace{1ex} \rho(f) = P_{i+j}\mbox{ and } \rho(g)=P_{i-j}.\]
Another representation is obtained by setting
\[ \rho(a) = P_i,\hspace{1ex} \rho(b) = P_j,\hspace{1ex} \rho(c) = P_k,\hspace{1ex} \rho(d) = P_{j},\hspace{1ex} \rho(e) = P_{k},\hspace{1ex} \rho(f) = P_{j}\mbox{ and } \rho(g)=P_{i}.\]
Of course, this second representation does not embed the vertex set into the projections, but this is not required. A further representation can be found even in $\mathbb{C}^n$ for any $n\geq 1$. Here we use $0$ and $1$ for projections onto the zero subspace and the whole space. Set $\rho(a)=\rho(c)=\rho(e)=\rho(g)=0$ and $\rho(b)=\rho(d)=\rho(f)=1$.
\end{ex}

A principle contribution of \cite{Fritz} is to provide a translation between quantum solutions of linear equations over $\mathbb{Z}_2$ and quantum representations of hypergraphs. The key result is the following \cite[Lem.~10]{Fritz}.

\begin{thm}
There is an algorithm to compute, for every linear system $Mx=b$, a finite hypergraph so that the quantum solutions of the linear system are in bijective correspondence with the representations of the hypergraph.
\end{thm}

The first idea behind the translation is that symmetries of a Hilbert space $\mc{H}$, i.e. bounded self-adjoint operators $A$ with $A^2=1$, are in bijective correspondence with projections of $\mc{H}$. A symmetry $A$ yields the projection $\frac{1}{2}(1+A)$ and a projection $P$ gives the symmetry $2P-1$. In fact, a symmetry $A$ has spectral decomposition $\frac{1}{2}(1+A)-\frac{1}{2}(1-A)$ so corresponds to an experiment with 2 outcomes. Commutivity of symmetries corresponds to that of their associated projections. So some aspects of the translation are relatively straightforward, but a very clever construction of is required to match all the requirements. Combining this with Slofstra's result gives the following.

\begin{cor}
It is undecidable whether a finite hypergraph has a quantum representation.
\end{cor}

For a given hypergraph, it is a simple matter to encode the conditions for it to have a quantum representation as a conjunction of equations. For example, if $v,w$ appear in a common edge then we require the equation $\rho(v)=\rho(v)\wedge \rho(w)'$ giving their orthogonality, and so forth. Then taking the conjunction of these equations to imply $0=1$ gives a quasi-equation that is valid in $\mc{C}(\mc{H})$ for all $\mc{H}$ of dimension greater than $0$ iff the hypergraph has no quantum representation. This yields the following result that was stated at the end of the previous section.

\begin{thm}
For an infinite-dimensional Hilbert space $\mc{H}$, the uniform word problem for the variety generated by $\mc{C}(\mc{H})$ is unsolvable.
\end{thm}

\begin{rem}
There is a connection between \ts{omp}s and hypergraphs. For simplicity, we restrict discussion to finite \ts{omp}s, but many things hold in a wider setting. There are several ways to attach a hypergraph to a finite \ts{omp} $P$. One way takes all elements of $P$ as vertices and uses as edges all pairwise orthogonal subsets of $P$ with join 1. Another way of doing this, an extension of ``Greechie diagrams'' \cite{Kalmbach} takes as vertices the atoms of $P$ and takes as edges those pairwise orthogonal sets of atoms with join 1. The hypergraph shown in Example~\ref{nh} is the 3-chain of Definition~\ref{aqw} considered as an \ts{omp}.
\end{rem}

One easily gains a feeling that the work on quantum solutions of linear systems and representations of hypergraphs has implications for the problem of embedding finite \ts{omp}s into $\mc{C}(\mc{H})$. But the situation is not so clear. Is it decidable whether members of the very special class of hypergraphs that arise from finite \ts{omp}s are representable? On the other hand, representability of the hypergraph of an \ts{omp} does not imply that it is embedable, embedability is equivalent to the existence of an injective representation. Never-the-less, this seems an area worthy of further study.

\begin{problem}
When can an \ts{omp} be embedded into an \ts{oml}?
\end{problem}

To reiterate the setting, an embedding $f:P\to L$ of an \ts{omp} $P$ into an \ts{oml} $L$ is a one-one map that preserves orthocomplementation and finite orthogonal joins. This implies that it preserves bounds and order. The following is given in \cite{harding1}.

\begin{thm}\label{nogo}
The \ts{omp} with Greechie diagram below cannot be embedded into an \ts{oml}.
\end{thm}

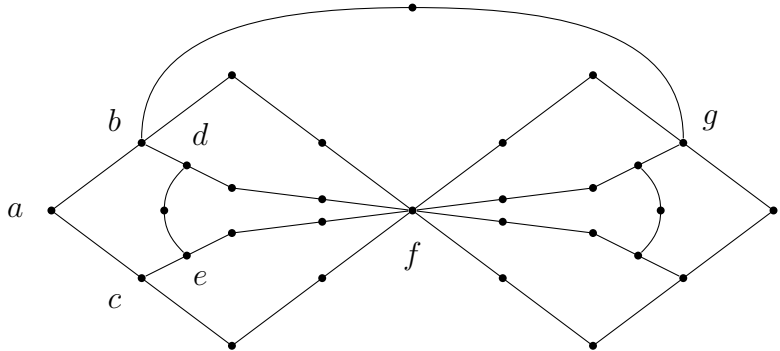
\begin{center}
\begin{tikzpicture}[scale=1.2,Q/.style={circle,fill=black,inner sep=0pt,minimum size=3pt}]
\node[Q] (e) at (0,0) {};
\node[Q] at (-1,.75) {}; \node[Q] at (-2,1.5) {}; \node[Q] at (-1,-.75) {}; \node[Q] at (-2,-1.5) {}; \node[Q] (b) at (-3,.75) {}; \node[Q] (a) at (-4,0) {}; \node[Q] (h) at (-3,-.75) {}; \draw (0,0)--(-2,1.5)--(-4,0)--(-2,-1.5)--(0,0);
\node[Q] at (-2,.25) {}; \node[Q] (w) at (-2.5,.5) {}; \node[Q] at (-1,.125) {};
\node[Q] at (-2,-.25) {}; \node[Q] (v) at (-2.5,-.5) {}; \node[Q] at (-1,-.125) {};
\draw (h)--(-2,-.25)--(e)--(-2,.25)--(b);
\node[Q] at (-2.75,0) {}; \draw (v) [out=135, in=-90] to (-2.75,0) [out=90, in=-135] to (w);
\node at (-3.3,-1) {$c$}; \node at (-4.4,0) {$a$}; \node at (-3.3,1) {$b$}; \node at (-2.35,-.75) {$e$}; \node at (-2.35,.85) {$d$};
\node[Q] at (1,.75) {}; \node[Q] at (2,1.5) {}; \node[Q] at (1,-.75) {}; \node[Q] at (2,-1.5) {}; \node[Q] (b') at (3,.75) {}; \node[Q] (a') at (4,0) {}; \node[Q] (h') at (3,-.75) {}; \draw (0,0)--(2,1.5)--(4,0)--(2,-1.5)--(0,0);
\node[Q] at (2,.25) {}; \node[Q] (w') at (2.5,.5) {}; \node[Q] at (1,.125) {};
\node[Q] at (2,-.25) {}; \node[Q] (v') at (2.5,-.5) {}; \node[Q] at (1,-.125) {};
\draw (h')--(2,-.25)--(e)--(2,.25)--(b');
\node[Q] at (2.75,0) {}; \draw (v') [out=45, in=-90] to (2.75,0) [out=90, in=-45] to (w');
\node at (3.3,1) {$g$};
\node at (0,-.5) {$f$}; \node[Q] at (0,2.25) {}; \draw (b) [out=90, in=180] to (0,2.25) [out=0, in=90] to (b');
\end{tikzpicture}
\end{center}
\vspace{2ex}

Note that this Greechie diagram has all edges with 3 elements.

\begin{proof}
Suppose that this \ts{omp} is embedded into an \ts{oml} $L$. We recall that the extended Foulis-Holland theorem \cite{distrib} says that if $S$ is a subset of an \ts{oml} and for any 3-element subset of $S$ there is one element that commutes with the other two, then the sublattice generated by $S$ is distributive. Using this and writing $+$ for join and (suppressed) multiplication for meet in $L$ we have

\begin{align*}
f\,&\leq\, (a+b)(a+c)(b+d)(c+e)\\
&=\,[(a+b)(a+c)]\, [(b+a)(b+d)]\, [(c+a)(c+e)] \\
&=\,(a+bc)(b+ad)(c+ae)\\
&=\,[(a+bc)(b+ad)]\, [(a+bc)(c+ae)]\\
&=\,[ab+bc+ad+abcd]\, [ac+bc+ae+abce]\\
&=\,(bc+ad)(bc+ae)\\
&=\,bc+ade\\
&=\,bc
\end{align*}

\noindent From this, we have $f\leq b$, and by symmetry, $f\leq g$, and as $b,g$ are orthogonal, $f=0$.
\end{proof}

We move to what at first seems an unrelated topic.

\begin{defn}
For a set $X$, an ordered pair $(\alpha,\alpha')$ of equivalence relations on $X$ is a {\em factor pair} if $\alpha\cap\alpha'$ is the diagonal relation $\Delta$ on $X$ and the relational product $\alpha\circ\alpha'$ is the universal relation $\nabla=X^2$ on $X$. Let Fact $(X)$ be the set of all factor pairs of $X$.
\end{defn}

The motivation behind factor pairs is that they encode the direct product decompositions of a set. Indeed, factor pairs are exactly the kernels of the projection operators associated to a direct product decomposition $X\simeq Y\times Z$. The following was established in \cite{harding1}.

\begin{thm}
Fact $(X)$ is an \ts{omp} with $0=(\Delta,\nabla), 1=(\nabla,\Delta)$; orthocomplementation given by $(\alpha,\alpha')^\perp = (\alpha',\alpha)$; and $(\alpha,\alpha')\leq(\beta,\beta')$ iff $\alpha\subseteq\beta$, $\beta'\subseteq\alpha'$, and all equivalence relations involved permute.
\end{thm}

As was first observed by Chin and Tarski \cite{Tarski}, under quite special circumstances, a small fragment of distributivity holds among equivalence relations \cite[Lem.~7.2]{harding1}. Using this, one can establish the following, where an embedding of one \ts{omp} into another is a one-one map that preserves orthocomplementation and finite orthogonal joins.

\begin{thm}\label{nogo1}
The \ts{omp} from Theorem~\ref{nogo} cannot be embedded into an \ts{omp} Fact~$(X)$ for any set $X$.
\end{thm}

The resemblance between Theorems~\ref{nogo} and \ref{nogo1} goes beyond their statements. Their proofs are nearly identical, using the fragment of distributivity that holds in relation algebras in place of the generalized Foulis-Holland theorem that provides a fragment of distributivity in \ts{oml}s. There are other finite \ts{omp}s that cannot be embedded into \ts{oml}s, and in each case, the proof of non-embedability into an \ts{oml} transfers transparently to a proof of non-embeddability into an \ts{omp} Fact~$(X)$. This raises a number of issues, such as whether Fact~$(X)$ can be embedded into an \ts{oml}, whether every \ts{oml} can be embedded as an \ts{omp} into some Fact~$(X)$, and whether there are further fragments of distributivity in a relation algebra to mirror the situation with the generalized Foulis-Holland theorem.






\newpage

\begin{footnotesize}

\end{footnotesize}

\end{document}